\documentclass[11pt,a4paper]{article}

\usepackage[utf8]{inputenc}
\usepackage[T1]{fontenc}
\usepackage{lmodern} 
\usepackage[english]{babel}
\usepackage{amsmath, amsfonts, amssymb, amsthm}
\usepackage{hyperref}
\usepackage{url}
\usepackage{complexity}
\usepackage{graphicx}

\usepackage[capitalise]{cleveref}

\title{Meta Theorem for Hardness on FCP{-}Problem}

\author{
  Atsuki Nagao\thanks{Faculty of Core Research Natural Science Division, Ochanomizu University, Tokyo, Japan. \texttt{a-nagao@is.ocha.ac.jp}. JSPS KAKENHI Grant Number 23K10981, and 24K02898.} 
  \and 
  Mei Sekiguchi\thanks{Graduate School of Humanities and Sciences, Ochanomizu University, Tokyo, Japan. \texttt{g2020529@is.ocha.ac.jp}.}
}

\date{} 

\newtheorem{theorem}{Theorem}
\newtheorem{lemma}[theorem]{Lemma}

\newtheorem{definition}[theorem]{Definition}
\newtheorem{corollary}[theorem]{Corollary}

\newcommand{\sigmatwop}{$\mathrm{\Sigma}_2^\P$}

\usepackage{abstract}

\begin{document}

\maketitle

\begin{abstract}
The Fewest Clues Problem (FCP) framework has been introduced to study the complexity of determining whether a solution to an \NP~problem can be uniquely identified by specifying a subset of the certificate. 
For a given problem $P \in \NP$, its FCP variant is denoted by FCP-$P$. 
While several \NP-complete problems have been shown to have \sigmatwop-complete FCP variants, it remains open whether this holds for all \NP-complete problems. 

In this work, we propose a meta-theorem that establishes the \sigmatwop-completeness of FCP-$P$ under the condition that the \NP-hardness of $P$ is proven via a polynomial-time reduction satisfying certain structural properties. 
Furthermore, we apply the meta-theorem to demonstrate the \sigmatwop-completeness of the FCP variants of several \NP-complete problems.
\end{abstract}

\section{introduction}
One of the major themes in theoretical computer science is understanding the computational complexity of determining whether a solution to a problem is unique.
This \textit{uniqueness problem} asks whether a given instance admits exactly one solution that satisfies its constraints.
A well-known example is the Unique $k$-SAT problem, whose complexity status—whether it is $\P^\NP$-complete—remains unresolved~\cite{blass1982unique}.

Closely related to this is the study of \textit{uniquification problem}, which investigates whether one can modify an instance so that it has a unique solution.
The Fewest Clues Problem (FCP)~\cite{Erik2018} provides a formal framework for this: given an instance of an \NP-complete problem and an integer $k$, the question is whether there exists a partial assignment of size $k$ (called \textit{clues}) that uniquely determines the solution.
It has been shown, for instance, that FCP-1-in-3 SAT is \sigmatwop-complete.

Here, \sigmatwop is the class of languages decidable by a polynomial-time nondeterministic Turing machine with access to an \NP~oracle.
In fact, the uniqueness versions of several classical problems—such as the Traveling Salesman Problem on undirected graphs, 
Integer Programming, and Knapsack—are known to be \sigmatwop-complete~\cite{papadimitriou2003computational}.
Moreover, it has been established that the FCP versions of certain \NP-complete problems are also \sigmatwop-complete~\cite{Erik2018, horiyama2024theoretical}.

These results suggest a general conjecture: for every \NP-complete problem, its FCP version is \sigmatwop-complete.
However, this conjecture has so far been addressed on a case-by-case basis.
No general methodology currently exists to establish the \sigmatwop-completeness of FCP versions for arbitrary \NP-complete problems.

A closely related framework is the Another Solution Problem (ASP)~\cite{yato2003complexity}, which asks whether a given instance has a solution different from a specified one.
\NP~problems can be shown to be ASP-complete using parsimonious reductions from ASP-hard problems.
In contrast, such reductions are not always sufficient to establish \sigmatwop-completeness for FCP problems.
In fact, there are known cases where the reduction used to prove \NP-completeness differs from that used to show \sigmatwop-completeness~\cite{sekiguchi2024}.

In this work, we propose a meta-theorem for the FCP framework that provides sufficient conditions—beyond parsimonious reductions—for establishing \sigmatwop-completeness.
In particular, our approach involves adding constraints on the structure of certificates across problems involved in the reduction.

\section{Preliminaries}
In this paper, we propose a meta-theorem that establishes \sigmatwop-hardness for the FCP version of an \NP-complete problem, provided that the problem satisfies specific conditions on the structure of its certificate mapping under a parsimonious reduction.
Our approach is based on analyzing how the certificates of two such problems relate to each other under a parsimonious polynomial-time reduction.
To establish our result, we identify structural conditions under which the uniqueness of solutions (as required in FCP) is preserved through the reduction.
In this section, we introduce the necessary definitions and terminology used throughout the paper.

\subsection{Fewest Clues Problem (FCP)}
The Fewest Clues Problem (FCP) was introduced in~\cite{Erik2018} as a framework to ask whether a partial assignment (called a clue) can be made so that a given \NP problem has a unique solution.
In the context of FCP, a \emph{clue} is defined in terms of partial information about a certificate.

\begin{definition}[\cite{Erik2018}]
Let $A$ be a problem in \NP, and let $x = (x_i) \in \Sigma^*$ be a certificate string (i.e., a solution to an instance $I$ of $A$). 
A string $c = (c_i) \in (\Sigma \cup \{\bot\})^*$ is called a \emph{clue} if there exists a certificate string $x$ such that $c_i = x_i$ whenever $c_i \ne \bot$. 
The size of a clue is defined as the number of positions $i$ such that $c_i \ne \bot$. 
We say that $x$ \emph{satisfies} the clue $c$, and write $c \subset x$.
\end{definition}

Using this notion, the FCP version of a problem $A$ in \NP is defined as follows:

\begin{definition}[\cite{Erik2018}]
For a problem $A$ in \NP, \textup{FCP}-$A$ is the following decision problem:
\begin{quote}
Given an instance $I$ of $A$ and an integer $k$, does there exist a clue $c$ of size at most $k$ such that $I$ has a unique solution consistent with $c$?
\end{quote}
\end{definition}

\subsection{Parsimonious Reduction}
To discuss reductions in the context of FCP, we must consider whether reductions preserve the number of solutions.
We therefore focus on the notion of a \emph{parsimonious polynomial-time reduction}~\cite{papadimitriou2003computational}, which is stronger than general polynomial-time reductions.

\begin{definition}
A reduction from a problem $P$ to a problem $Q$ consists of two parts:
a function $R$ that maps instances $I_P$ of $P$ to instances $I_Q=R(I_P)$ of $Q$, and a function $\mathcal{S}$ that maps any solution $s_Q$ of $I_Q$ to a solution $s_P=\mathcal{S}(s_Q)$ of $I_P$.
The reduction is said to be \emph{parsimonious} if $\mathcal{S}$ is the identity function; that is, the number of solutions of $R(I_P)$ is exactly the same as the number of solutions of $x$.
\end{definition}

Throughout this paper, unless otherwise stated, all reductions are assumed to be parsimonious and polynomial-time computable.

\begin{figure}[tb]
    \centering    
    \includegraphics[scale=0.1]{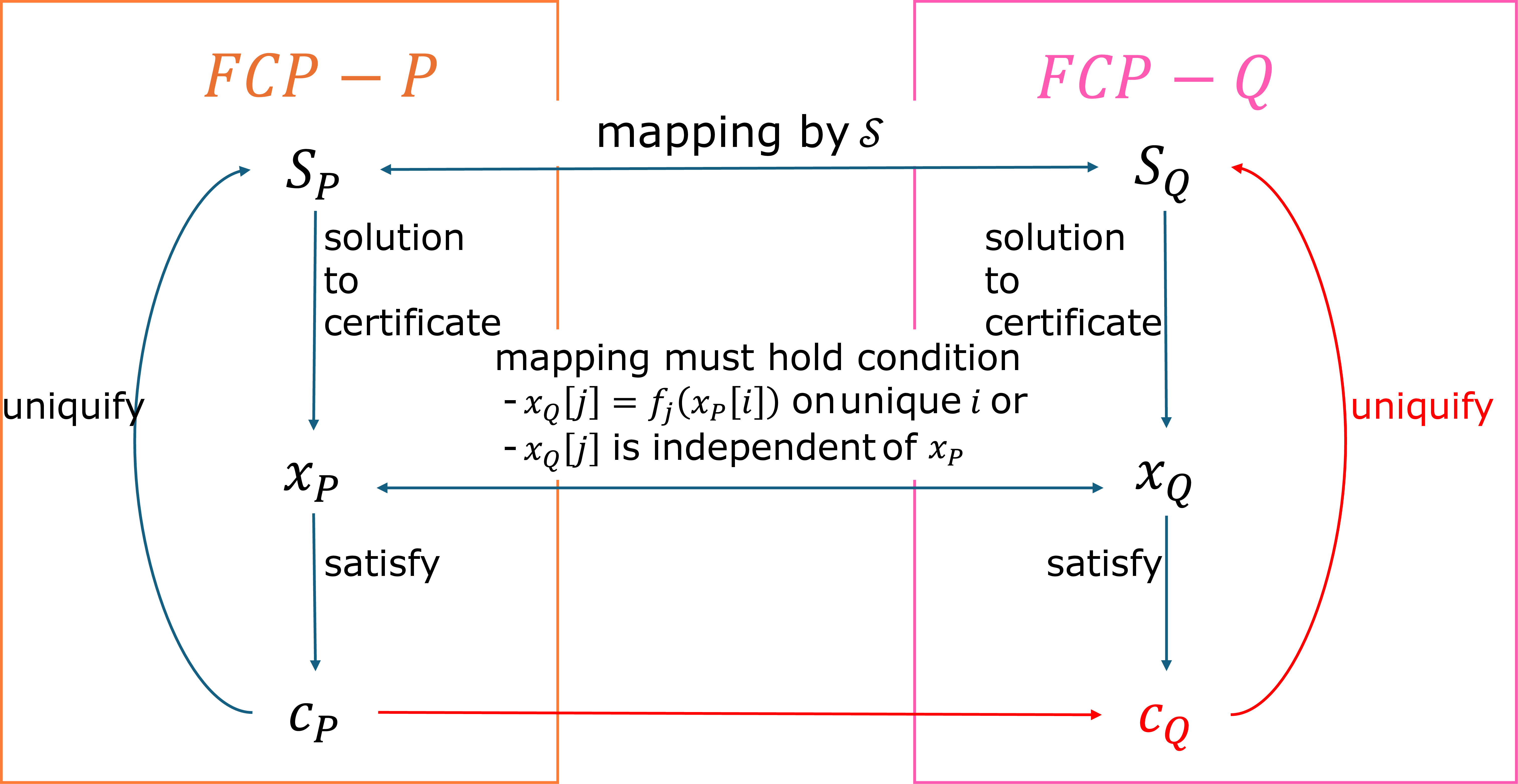}
    \caption{Overview of variables related to FCP-$P$ and FCP-$Q$ in Theorem~\ref{thm:meta}. The conditions and relationships among $c_P$, $c_Q$, and others are assumptions of the theorem. The conclusion is that $c_Q$ exists and can uniquify $S_Q$ using $c_P$ and the parsimonious reduction.}
    \label{fig:overall}
\end{figure}

\subsection{Meta-Theorem}
In this section, we introduce our main result: a meta-theorem for proving the \sigmatwop-completeness of FCP problems using parsimonious reductions and structural assumptions on certificate mappings.

Let $P$ be an \NP-complete problem whose FCP version, FCP-$P$, is known to be \sigmatwop-complete. 
Let $Q$ be a problem such that there exists a parsimonious polynomial-time reduction $R$ from $P$ to $Q$.
Let $I_P$ and $I_Q$ be instances of $P$ and $Q$, respectively, and let $S_P$ and $S_Q$ denote their respective solutions (when the instances are Yes-instances).
Let $x_P \in \Gamma_P^m$ and $x_Q \in \Gamma_Q^n$ be certificate strings for $S_P$ and $S_Q$, respectively.
Let $c_P$ be a clue for FCP-$P$, and $c_Q$ be a clue for FCP-$Q$.

We now state our meta-theorem.

\begin{theorem}\label{thm:meta}
Assume the following conditions:
\begin{itemize}
    \item $P$ is an \NP-complete problem, and FCP-$P$ is \sigmatwop-complete.
    \item There exists a parsimonious polynomial-time reduction $R$ from $P$ to $Q$.
    \item $x_P \in \Gamma_P^m$ is a certificate for the solution $S_P$ of instance $I_P$.
    \item $x_Q \in \Gamma_Q^n$ is a certificate for the solution $S_Q$ of instance $I_Q$ corresponding to $S_P$ via $R$.
\end{itemize}

If the following condition holds for every $j = 0, \dots, n-1$, then FCP-$Q$ is \sigmatwop-complete:

For each $j$, one of the following must be satisfied:
\begin{itemize}
    \item There exists an injective function $f_j: \Gamma_P \rightarrow \Gamma_Q$ and a unique $i$ such that $x_Q[j] = f_j(x_P[i])$.
    \item $x_Q[j]$ is constant with respect to $x_P$.
\end{itemize}
\end{theorem}

\begin{figure}[hbt]
    \centering    
    \includegraphics[scale=0.082]{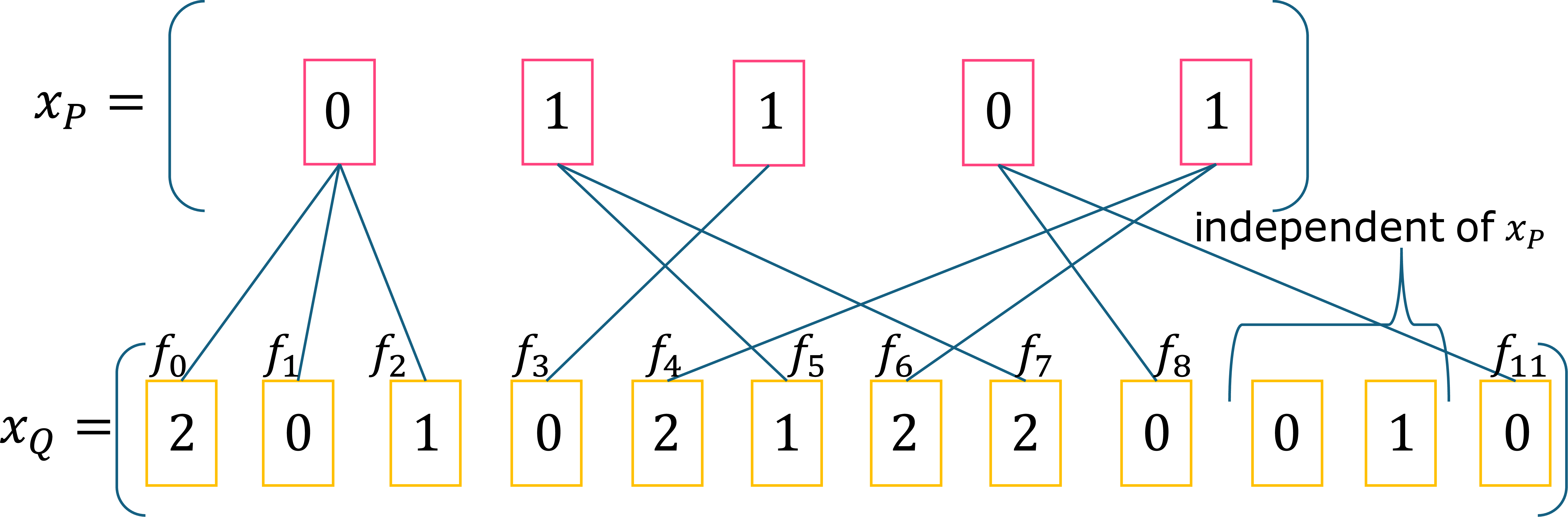}
    \caption{An example of $x_P$ and $x_Q$ satisfying the conditions of Theorem~\ref{thm:meta}. Each character of $x_Q$ depends on exactly one character of $x_P$ via an injective function $f_j$.}
    \label{fig:okexample}
\end{figure}

\begin{figure}[tb]
    \centering    
    \includegraphics[scale=0.082]{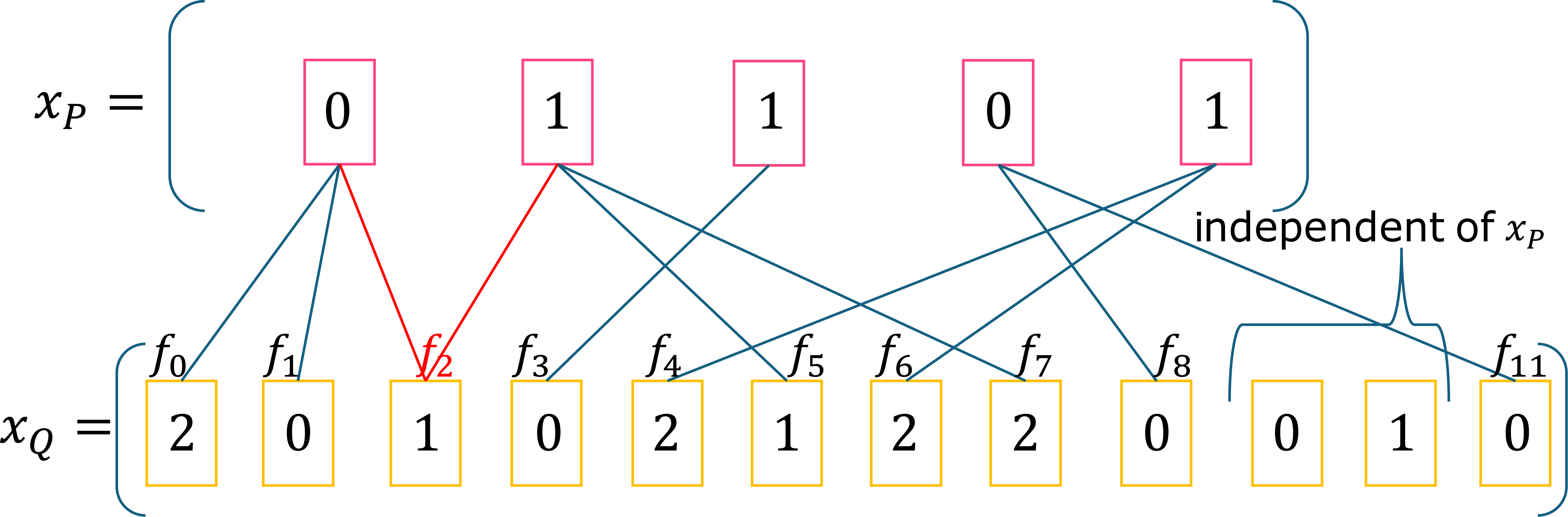}
    \caption{An example of $x_P$ and $x_Q$ that \textbf{does not} satisfy the conditions of Theorem~\ref{thm:meta}. Here, $x_Q[2]$ depends on both $x_P[0]$ and $x_P[1]$, violating the uniqueness condition.}
    \label{fig:ngexample}
\end{figure}

Figure~\ref{fig:okexample} illustrates a case where the conditions of the theorem are satisfied: each character in \( x_Q \) is either determined by an injective function depending on exactly one character in \( x_P \), or is a constant independent of \( x_P \).
As illustrated in Figure~\ref{fig:ngexample}, if a character in $x_Q$ depends on multiple characters in $x_P$, then the condition of uniqueness is violated.
Similarly, if any function $f_j$ is not injective, the theorem does not apply.

\subsection{Terminology}

To support the proof of Theorem~\ref{thm:meta}, we summarize the key terms used in the context of problems $P$ and $Q$, their FCP variants, and the associated instances and clues:
\begin{itemize}
    \item $I_P$ is an instance of $P$, and $S_P$ is a solution to $I_P$ if it is a Yes-instance.
    \item $I_Q$ is an instance of $Q$, and $S_Q$ is a solution to $I_Q$ if it is a Yes-instance.
    \item $c_P$ is a clue that uniquifies $S_P$ for the instance $(I_P, k)$ of FCP-$P$.
    \item $c_Q$ is a clue that uniquifies $S_Q$ for the instance $(I_Q, k)$ of FCP-$Q$.
\end{itemize}

\section{Proof of the Theorem}
In this section, we prove Theorem~\ref{thm:meta}.  
Under the assumptions stated in the theorem, we construct a reduction from FCP-$P$ to FCP-$Q$.

\begin{proof}
Let $I_P$ be a yes-instance of $P$, and let $I_Q$ be the corresponding instance of $Q$ under the reduction $R$.  
Let $S_P$ be the unique solution to $I_P$, and $S_Q$ be the corresponding solution to $I_Q$.  
Let $x_P \in \Gamma_P^m$ be the certificate sequence representing $S_P$, and $x_Q \in \Gamma_Q^n$ be the certificate sequence representing $S_Q$.

We assume that for any solution $S_P$ and for each $j = 0, \dots, n-1$, one of the following holds:
\begin{itemize}
    \item There exists an injective function $f_j : \Gamma_P \to \Gamma_Q$ and a unique $i$ such that $x_Q[j] = f_j(x_P[i])$.
    \item $x_Q[j]$ is independent of $x_P$.
\end{itemize}

Each index $j$ of $x_Q$ can thus be categorized into at most $m+1$ groups, depending on which index of $x_P$ it depends on, or whether it is independent.  
The value of each such $x_Q[j]$ is either determined by a specific index of $x_P$ via some $f_j$, or is fixed regardless of $x_P$.

We now prove the following two lemmas, which establish the correctness of the reduction.

\begin{lemma} \label{lem:yes2yes}
Given $I_P$ and an integer $k$,  
if there exists a clue $c_P$ of size at most $k$ that uniquely determines the solution to $I_P$,  
then there exists a clue $c_Q$ of size at most $k$ that uniquely determines the solution to $I_Q$.
\end{lemma}

\begin{proof}
Let $L$ be the set of indices in $x_P$ at which $c_P$ assigns non-$\bot$ values.  
We construct a clue $c_Q \in \Gamma_Q^n$ that uniquely determines the solution to $I_Q$ using $f_0, \dots, f_{n-1}$, $L$, and $x_Q$.

We initialize $c_Q = \bot^n$.
For each index $j$ such that $x_Q[j]$ depends on $x_P[i]$ via $f_j$,
we update $c_Q[j] = f_j(c_P[i])$ only if $i \in L$ and $j$ is the smallest index such that $f_j$ maps from $x_P[i]$.
Otherwise, $c_Q[j]$ remains $\bot$.
For indices $j$ independent of $x_P$, $c_Q[j]$ also remains $\bot$.

Figures~\ref{fig:clue1} through~\ref{fig:clue5} (they are on the last page of this manusucript) illustrate the partitioning of index sets and the update process for $c_Q$.

We now prove that the constructed $c_Q$ uniquely determines the solution to $I_Q$.  
Suppose, for contradiction, that it does not.  
Then there exist two distinct solutions $S_Q$ and $S_Q'$ to $I_Q$ that both satisfy $c_Q$.  
By the existence of $R$, there exist corresponding solutions $S_P$ and $S_P'$ to $I_P$ for $S_Q$ and $S_Q'$, respectively.  
Since $S_Q \neq S_Q'$, it follows that $S_P \neq S_P'$, and both satisfy $c_P$ by.
This contradicts the assumption that $c_P$ uniquely determines $S_P$.

Therefore, $c_Q$ uniquely determines the solution to $I_Q$.
\end{proof}

\begin{lemma} \label{lem:no2no}
Given $I_Q$ and an integer $k$,  
if there exists a clue $c_Q$ of size at most $k$ that uniquely determines the solution to $I_Q$,  
then there exists a clue $c_P$ of size at most $k$ that uniquely determines the solution to $I_P$.
\end{lemma}

\begin{proof}
We consider the equivalence classes of indices of $x_Q$, where each class consists of indices that depend on the same index $i$ of $x_P$ via some $f_j$, or are independent of $x_P$.
Let us examine each such class individually.
If a class contains any index $j$ such that $c_Q[j] \neq \bot$, and this $j$ depends on some $x_P[i]$ via injective function $f_j$, then we set $c_P[i] = f_j^{-1}(c_Q[j])$.
Since $f_j$ is injective, this value is uniquely determined.
For all indices $x_P[i]$ that correspond to classes where all associated $c_Q[j] = \bot$, we assign $c_P[i] = \bot$.
In this way, we construct a clue $c_P$ of size at most $k$.

To prove that $c_P$ uniquely determines the solution to $I_P$, suppose for contradiction that there exist two distinct solutions $S_P$ and $S_P'$ that both satisfy $c_P$.  
By the existence of $R$, there exist corresponding solutions $S_Q$ and $S_Q'$ to $I_Q$.  
Then $S_Q \neq S_Q'$, and both satisfy $c_Q$, which contradicts the assumption that $c_Q$ uniquely determines $S_Q$.

Therefore, $c_P$ uniquely determines the solution to $I_P$.
\end{proof}

From Lemmas~\ref{lem:yes2yes} and~\ref{lem:no2no}, we conclude that FCP-$Q$ admits a clue of size at most $k$ if and only if FCP-$P$ does.  
Thus, the reduction from FCP-$P$ to FCP-$Q$ is sound and complete.  
This reduction runs in polynomial time, as it uses only the computation of $R$ and the structural mappings $f_j$.  
Hence, FCP-$Q$ is \sigmatwop-complete.
\end{proof}

\section{Applications of the Meta-Theorem}

By applying Theorem~\ref{thm:meta} to known reductions that establish \NP-hardness, we can prove the \sigmatwop-completeness of various FCP problems.  
The following are some examples.

\begin{corollary}
    \textup{FCP-\textsc{Yosenabe}} is \sigmatwop-complete. 
\end{corollary}
\begin{proof} 
    The reduction presented in \cite{iwamoto2014yosenabe}, which establishes the \NP-hardness of \textsc{Yosenabe}, satisfies the conditions of Theorem~\ref{thm:meta}.
\end{proof}

\begin{corollary}
    \textup{FCP-\textsc{Choco Banana}} is \sigmatwop-complete. 
\end{corollary}
\begin{proof}
    The reduction presented in \cite{iwamoto2024choco}, which establishes the \NP-hardness of \textsc{Choco Banana}, satisfies the conditions of Theorem~\ref{thm:meta}.
\end{proof} 

\begin{corollary}
    \textup{FCP-\textsc{Nondango}} is \sigmatwop-complete. 
\end{corollary}
\begin{proof}
    The reduction presented in \cite{ruangwises2023nondango}, which establishes the \NP-hardness of \textsc{Nondango}, satisfies the conditions of Theorem~\ref{thm:meta}.
\end{proof}

\begin{corollary}
    \textup{FCP-\textsc{Double Choco}} is \sigmatwop-complete. 
\end{corollary}
\begin{proof}
    The reduction presented in \cite{djuric2022double}, which establishes the \NP-hardness of \textsc{Double Choco}, satisfies the conditions of Theorem~\ref{thm:meta}.
\end{proof}

\section{Conclusion and Future Work}

In this study, we proposed and proved a meta-theorem that establishes the \sigmatwop-completeness of FCP problems derived from \NP-complete problems.  
Theorem~\ref{thm:meta} relies on preserving the clue size $k$ between instances under a parsimonious polynomial-time reduction.  
However, the requirement of preserving clue size is quite restrictive in practice.

One direction for future work is to relax these conditions so that \sigmatwop-hardness can be shown even when the clue sizes in the reduction differ between FCP problems.

Furthermore, the requirement in Theorem~\ref{thm:meta} that some entries in $x_Q$ are constant and independent of $x_P$ arises due to the nature of parsimonious reductions.  
In practice, many reductions (even non-parsimonious ones) produce $x_Q[i]$ values that are injectively determined from $x_P$, but do not contain truly constant values.  
In such cases, the uniqueness of the solution to the corresponding FCP instance cannot be directly guaranteed.  



\bibliography{metatheorem}

\begin{figure}[ht]
    \centering    
    \includegraphics[scale=0.082]{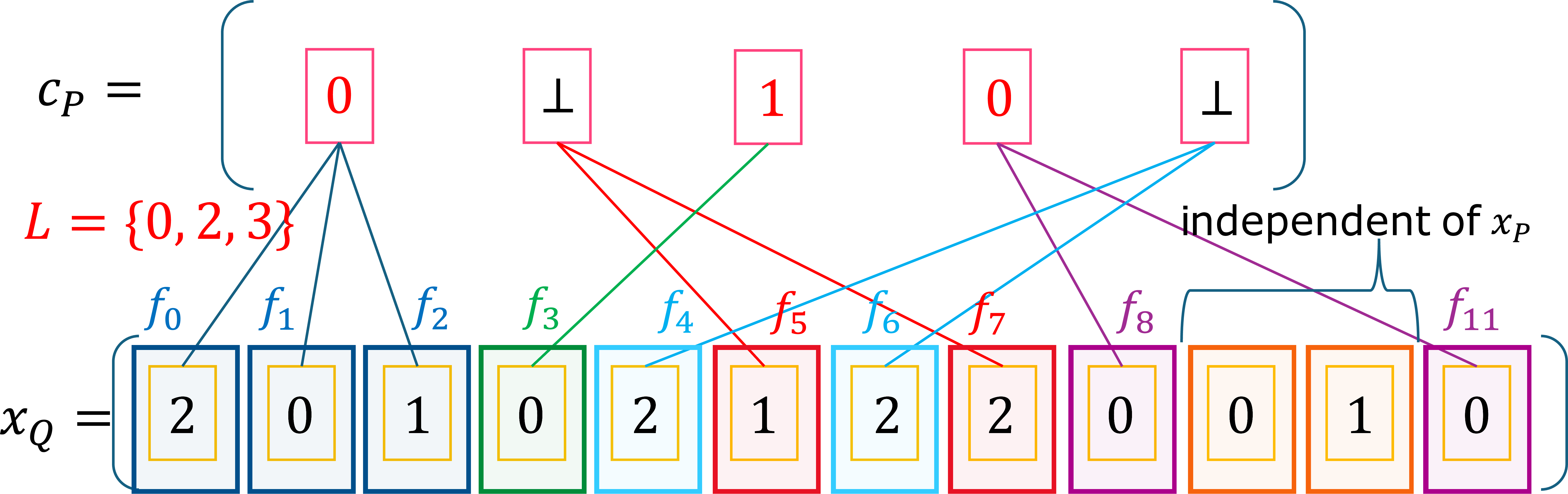}
    \caption{Given a clue \( c_P \subset x_P \) as shown, the set of specified indices is \( L = \{0, 2, 3\} \).}
    \label{fig:clue1}
\end{figure}

\begin{figure}[ht]
    \centering    
    \includegraphics[scale=0.082]{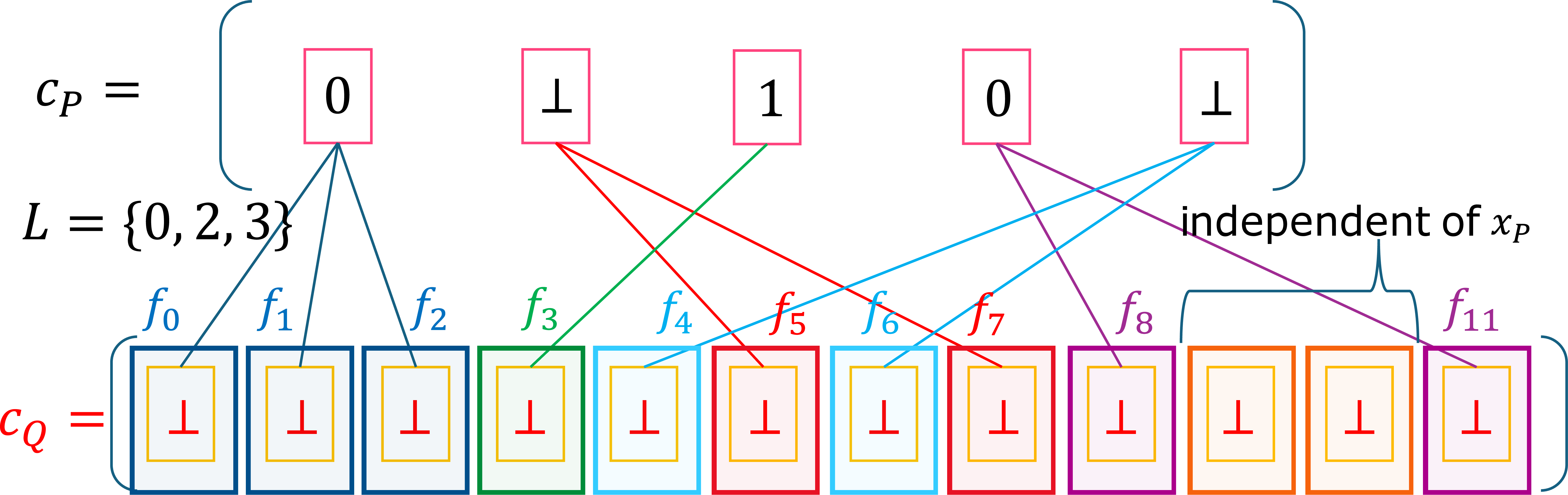}
    \caption{Initialization of clue \( c_Q \): fill all entries with \( \bot \).}
    \label{fig:clue2}
\end{figure}

\begin{figure}[ht]
    \centering    
    \includegraphics[scale=0.082]{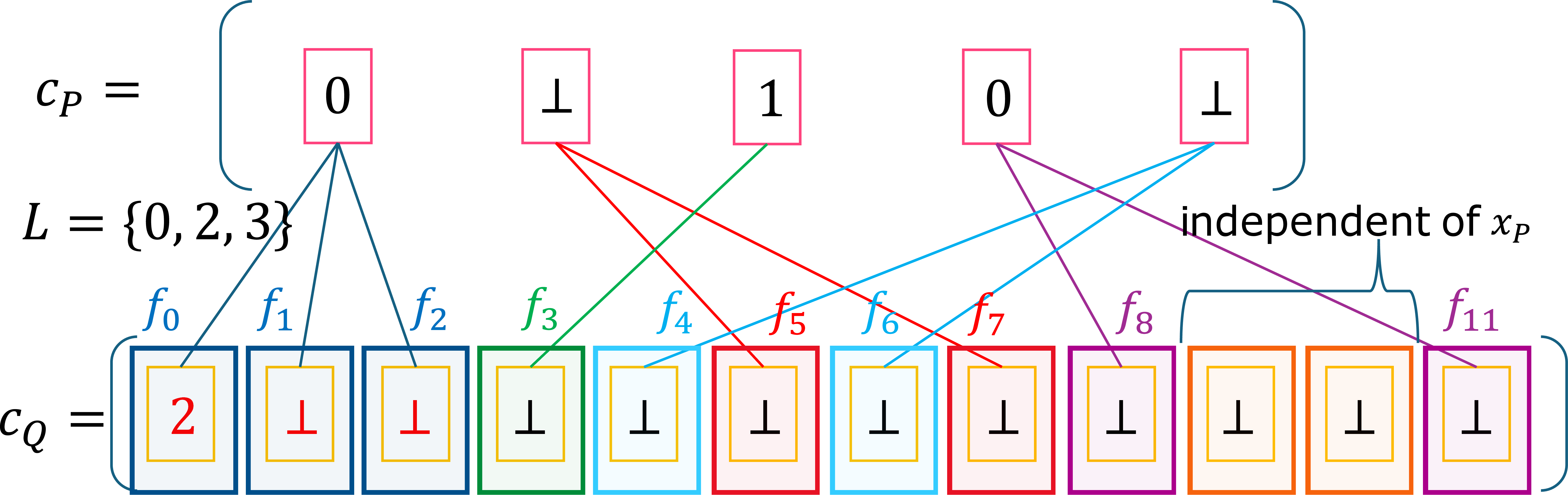}
    \caption{For each \( j \) such that \( x_Q[j] \) is determined by some \( f_j \) and the corresponding index \( i \in L \), update only the smallest such \( j \) by setting \( c_Q[j] = f_j(c_P[i]) \). Leave the rest as \( \bot \).}
    \label{fig:clue3}
\end{figure}

\begin{figure}[ht]
    \centering    
    \includegraphics[scale=0.082]{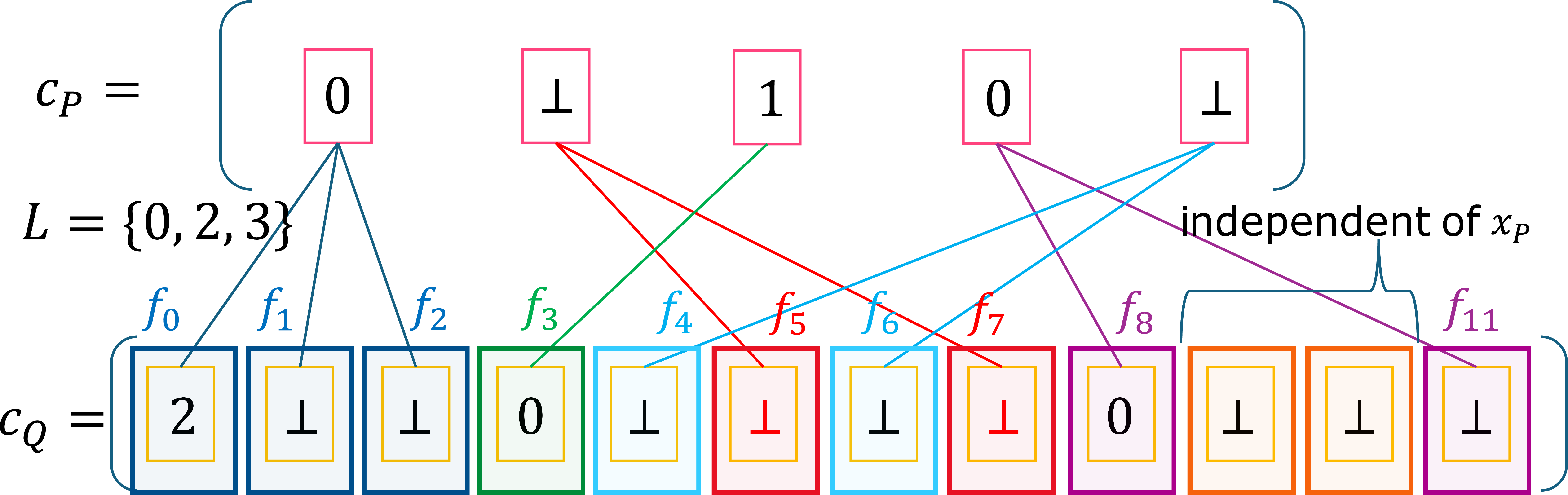}
    \caption{For each \( j \) such that \( x_Q[j] \) is determined by \( x_P[i] \) with \( i \notin L \), leave \( c_Q[j] \) unchanged as \( \bot \).}
    \label{fig:clue4}
\end{figure}

\begin{figure}[ht]
    \centering    
    \includegraphics[scale=0.082]{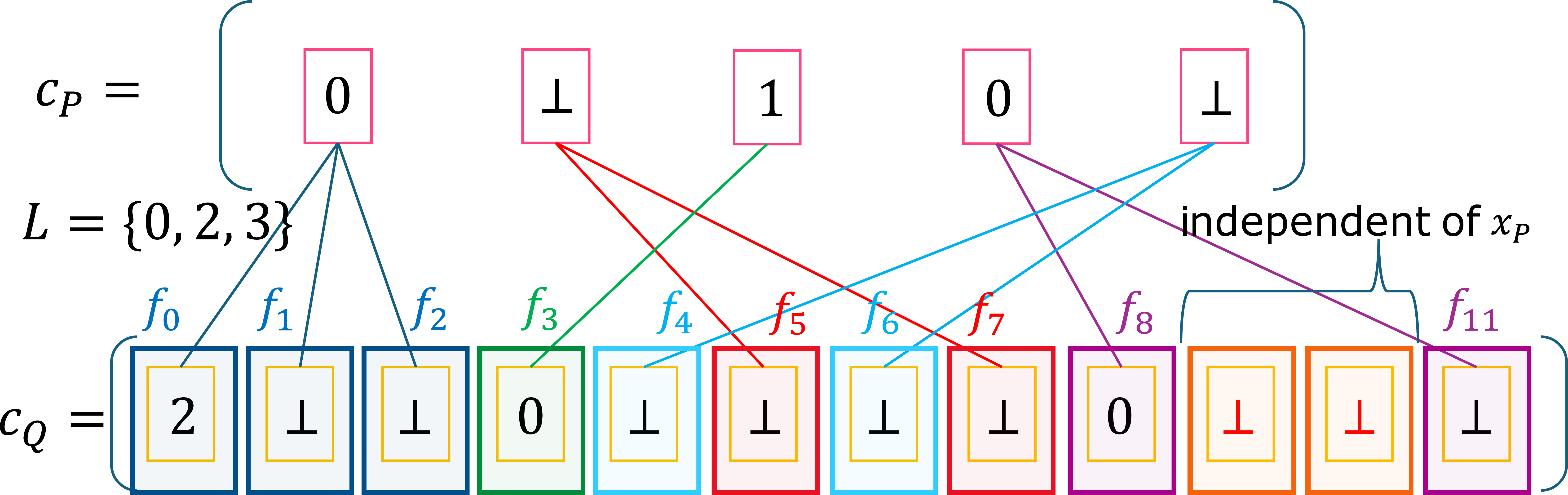}
    \caption{For each \( j \) such that \( x_Q[j] \) is independent of \( x_P \), leave \( c_Q[j] \) unchanged as \( \bot \).}
    \label{fig:clue5}
\end{figure}
\end{document}